\newif\ifdraft  \draftfalse   
\documentclass[12pt]{article}
\usepackage[latin1]{inputenc}
\usepackage[T1]{fontenc}
\usepackage{etoolbox}

\usepackage[margin=3cm,bottom=3cm,footskip=1.25cm]{geometry}
\usepackage{xspace}
\usepackage{amsthm}
\usepackage{amsmath,amsfonts,amscd,amssymb,amstext,dsfont}
\usepackage[mathlines]{lineno}

\usepackage{graphicx}
\usepackage{bold-extra}
\usepackage{vaucanson-g-O}
\usepackage{enumitem}
\usepackage{caption}
\usepackage{subcaption}
\usepackage{fancybox}
\usepackage{xstring}
\usepackage{subeqnarray}
\usepackage{float}
\usepackage{mparhack}
\usepackage{datetime}
\usepackage{everypage}
\AddEverypageHook{\reversemarginpar}

\newdateformat{isodate}{\THEYEAR\xmd-\twodigit{\THEMONTH}\xmd-\xmd\!\xmd\twodigit{\THEDAY}}
\isodate
\restylefloat{figure}\setlist{nolistsep}
\usepackage{ifthen}
\usepackage{etoolbox}
\usepackage{xargs}
\usepackage{slantsc}

\newcommandx{\newtheoremy}[3][2={}]{
  \ifthenelse{\equal{#2}{}}{
    \ifcsmacro{#1}{}{\newtheorem{#1}{#3}}
  }{
    \ifcsmacro{#1}{}{\newtheorem{#1}[#2]{#3}}
  }
}

\newtheoremy{thm}{thm}
\newcommand{\thmBlockFont}[1]{\sc{#1}}
\newcommand{\RefSeparator}{.}
\newcommand{\generalref}[2]{\ref{#1\RefSeparator#2}}
\newcommand{\generalpageref}[2]{\pageref{#1\RefSeparator#2}}
\newcommand{\generallabel}[2]{\label{#1\RefSeparator#2}}
\newcommand{\PageName}{page}

\newcommand{\AlgorithmName}{Algorithm}
\newcommand{\AlgorithmRefName}{Algorithm}
\newcommand{\AlgorithmRefPrefix}{a}
\newtheoremy{algorithm}[thm]{\thmBlockFont{\AlgorithmName}}

\makeatletter
\newcommand*{\ralgorithm}{\@ifstar{\generalref{\AlgorithmRefPrefix}}{\AlgorithmRefName~\ralgorithm*}}
\newcommand*{\palgorithm}{\@ifstar{\generalpageref{\AlgorithmRefPrefix}}{\PageName~\palgorithm*}}
\makeatother

\newcommand{\CorollaryName}{Corollary}
\newcommand{\CorollaryRefName}{Corollary}
\newcommand{\CorollaryRefPrefix}{c}
\newtheoremy{corollary}[thm]{\thmBlockFont{\CorollaryName}}

\makeatletter
\newcommand*{\rcorollary}{\@ifstar{\generalref{\CorollaryRefPrefix}}{\CorollaryRefName~\rcorollary*}}
\newcommand*{\pcorollary}{\@ifstar{\generalpageref{\CorollaryRefPrefix}}{\PageName~\pcorollary*}}
\makeatother

\newcommand{\ConjectureName}{Conjecture}
\newcommand{\ConjectureRefName}{Conjecture}
\newcommand{\ConjectureRefPrefix}{cj}
\newtheoremy{conjecture}[thm]{\thmBlockFont{\ConjectureName}}

\makeatletter
\newcommand*{\rconjecture}{\@ifstar{\generalref{\ConjectureRefPrefix}}{\ConjectureRefName~\rconjecture*}}
\newcommand*{\pconjecture}{\@ifstar{\generalpageref{\ConjectureRefPrefix}}{\PageName~\pconjecture*}}
\makeatother

\newcommand{\DefinitionName}{Definition}
\newcommand{\DefinitionRefName}{Definition}
\newcommand{\DefinitionRefPrefix}{d}
\newtheoremy{definition}[thm]{\thmBlockFont{\DefinitionName}}
\newcommand{\ldefinition}[1]{\generallabel{\DefinitionRefPrefix}{#1}}
\makeatletter
\newcommand*{\rdefinition}{\@ifstar{\generalref{\DefinitionRefPrefix}}{\DefinitionRefName~\rdefinition*}}
\newcommand*{\pdefinition}{\@ifstar{\generalpageref{\DefinitionRefPrefix}}{\PageName~\pdefinition*}}
\makeatother

\newcommand{\ExampleName}{Example}
\newcommand{\ExampleRefName}{Example}
\newcommand{\ExampleRefPrefix}{e}
\newtheoremy{example}[thm]{\thmBlockFont{\ExampleName}}
\newcommand{\lexample}[1]{\generallabel{\ExampleRefPrefix}{#1}}
\makeatletter
\newcommand*{\rexample}{\@ifstar{\generalref{\ExampleRefPrefix}}{\ExampleRefName~\rexample*}}
\newcommand*{\pexample}{\@ifstar{\generalpageref{\ExampleRefPrefix}}{\PageName~\pexample*}}
\makeatother

\newcommand{\LemmaName}{Lemma}
\newcommand{\LemmaRefName}{Lemma}
\newcommand{\LemmaRefPrefix}{l}
\newtheoremy{lemma}[thm]{\thmBlockFont{\LemmaName}}
\newcommand{\llemma}[1]{\generallabel{\LemmaRefPrefix}{#1}}
\makeatletter
\newcommand*{\rlemma}{\@ifstar{\generalref{\LemmaRefPrefix}}{\LemmaRefName~\rlemma*}}
\newcommand*{\plemmam}{\@ifstar{\generalpageref{\LemmaRefPrefix}}{\PageName~\plemma*}}
\makeatother

\newcommand{\PropositionName}{Proposition}
\newcommand{\PropositionRefName}{Proposition}
\newcommand{\PropositionRefPrefix}{p}
\newtheoremy{proposition}[thm]{\thmBlockFont{\PropositionName}}
\newcommand{\lproposition}[1]{\generallabel{\PropositionRefPrefix}{#1}}
\makeatletter
\newcommand*{\rproposition}{\@ifstar{\generalref{\PropositionRefPrefix}}{\PropositionRefName~\rproposition*}}
\newcommand*{\pproposition}{\@ifstar{\generalpageref{\PropositionRefPrefix}}{\PageName~\pproposition*}}
\makeatother

\newcommand{\PropertyName}{Property}
\newcommand{\PropertyRefName}{Property}
\newcommand{\PropertyRefPrefix}{pp}
\newtheoremy{property}[thm]{\thmBlockFont{\PropertyName}}

\makeatletter
\newcommand*{\rproperty}{\@ifstar{\generalref{\PropertyRefPrefix}}{\PropertyRefName~\rproperty*}}
\newcommand*{\pproperty}{\@ifstar{\generalpageref{\PropertyRefPrefix}}{\PageName~\pproperty*}}
\makeatother

\newcommand{\QuestionName}{Question}
\newcommand{\QuestionRefName}{Question}
\newcommand{\QuestionRefPrefix}{q}
\newtheoremy{question}[thm]{\thmBlockFont{\QuestionName}}

\makeatletter
\newcommand*{\rquestion}{\@ifstar{\generalref{\QuestionRefPrefix}}{\QuestionRefName~\rquestion*}}
\newcommand*{\pquestion}{\@ifstar{\generalpageref{\QuestionRefPrefix}}{\PageName~\pquestion*}}
\makeatother

\newcommand{\RemarkName}{Remark}
\newcommand{\RemarkRefName}{Remark}
\newcommand{\RemarkRefPrefix}{r}
\newtheoremy{remark}[thm]{\thmBlockFont{\RemarkName}}

\makeatletter
\newcommand*{\rremark}{\@ifstar{\generalref{\RemarkRefPrefix}}{\RemarkRefName~\rremark*}}
\newcommand*{\premark}{\@ifstar{\generalpageref{\RemarkRefPrefix}}{\PageName~\premark*}}
\makeatother

\newcommand{\NotationName}{Notation}
\newcommand{\NotationRefName}{Notation}
\newcommand{\NotationRefPrefix}{n}
\newtheoremy{notation}[thm]{\thmBlockFont{\NotationName}}

\makeatletter
\newcommand*{\rnotation}{\@ifstar{\generalref{\NotationRefPrefix}}{\NotationRefName~\rnotation*}}
\newcommand*{\pnotation}{\@ifstar{\generalpageref{\NotationRefPrefix}}{\PageName~\pnotation*}}
\makeatother

\newcommand{\TheoremName}{Theorem}
\newcommand{\TheoremRefName}{Theorem}
\newcommand{\TheoremRefPrefix}{t}
\newtheoremy{theorem}[thm]{\thmBlockFont{\TheoremName}}
\newtheoremy{Theorem}{\thmBlockFont{\TheoremName}}

\newcommand{\ltheorem}[1]{\generallabel{\TheoremRefPrefix}{#1}}
\makeatletter
\newcommand*{\rtheorem}{\@ifstar{\generalref{\TheoremRefPrefix}}{\TheoremRefName~\rtheorem*}}
\newcommand*{\ptheorem}{\@ifstar{\generalpageref{\TheoremRefPrefix}}{\PageName~\ptheorem*}}
\makeatother

\newcommand{\FigureRefName}{Figure}
\newcommand{\FigureRefPrefix}{f}
\newcommand{\lfigure}[1]{\generallabel{\FigureRefPrefix}{#1}}
\makeatletter
\newcommand*{\rfigure}{\@ifstar{\generalref{\FigureRefPrefix}}{\FigureRefName~\rfigure*}}
\newcommand*{\pfigure}{\@ifstar{\generalpageref{\FigureRefPrefix}}{\PageName~\pfigure*}}
\makeatother

\newcommand{\EquationRefName}{Equation}
\newcommand{\EquationRefPrefix}{eq}

\makeatletter
\newcommand*{\requation}{\@ifstar{\generalref{\EquationRefPrefix}}{\EquationRefName~\requation*}}
\newcommand*{\pequation}{\@ifstar{\generalpageref{\EquationRefPrefix}}{\PageName~\pequation*}}
\makeatother

\newcommand{\SectionRefName}{Section}
\newcommand{\SectionRefPrefix}{s}
\newcommand{\lsection}[1]{\generallabel{\SectionRefPrefix}{#1}}
\makeatletter
\newcommand*{\rsection}{\@ifstar{\generalref{\SectionRefPrefix}}{\SectionRefName~\rsection*}}
\newcommand*{\psection}{\@ifstar{\generalpageref{\SectionRefPrefix}}{\PageName~\psection*}}
\makeatother

\newcommand{\ProblemName}{Problem}
\newcommand{\ProblemRefName}{Problem}
\newcommand{\ProblemRefPrefix}{pb}
\newtheoremy{problem}[thm]{\thmBlockFont{\ProblemName}}
\newcommand{\lproblem}[1]{\generallabel{\ProblemRefPrefix}{#1}}
\makeatletter
\newcommand*{\rproblem}{\@ifstar{\generalref{\ProblemRefPrefix}}{\ProblemRefName~\rproblem*}}
\newcommand*{\pproblem}{\@ifstar{\generalpageref{\ProblemRefPrefix}}{\PageName~\pproblem*}}
\makeatother

\newcommand{\vmrefprefix}[1]{%
  \ifthenelse{\equal{#1}{corollary}}{\CorollaryRefPrefix}{}%
  \ifthenelse{\equal{#1}{definition}}{\DefinitionRefPrefix}{}%
  \ifthenelse{\equal{#1}{example}}{\ExampleRefPrefix}{}%
  \ifthenelse{\equal{#1}{lemma}}{\LemmaRefPrefix}{}%
  \ifthenelse{\equal{#1}{proposition}}{\PropositionRefPrefix}{}%
  \ifthenelse{\equal{#1}{property}}{\PropertyRefPrefix}{}%
  \ifthenelse{\equal{#1}{question}}{\QuestionRefPrefix}{}%
  \ifthenelse{\equal{#1}{remark}}{\RemarkRefPrefix}{}%
  \ifthenelse{\equal{#1}{notation}}{\NotationRefPrefix}{}%
  \ifthenelse{\equal{#1}{theorem}}{\TheoremRefPrefix}{}%
  \ifthenelse{\equal{#1}{figure}}{\FigureRefPrefix}{}%
  \ifthenelse{\equal{#1}{equation}}{\EquationRefPrefix}{}%
  \ifthenelse{\equal{#1}{section}}{\SectionRefPrefix}{}%
}

\newcommand{\vmrefname}[1]{
  \ifthenelse{\equal{#1}{corollary}}{\CorollaryRefName}{}%
  \ifthenelse{\equal{#1}{definition}}{\DefinitionRefName}{}%
  \ifthenelse{\equal{#1}{example}}{\ExampleRefName}{}%
  \ifthenelse{\equal{#1}{lemma}}{\LemmaRefName}{}%
  \ifthenelse{\equal{#1}{proposition}}{\PropositionRefName}{}%
  \ifthenelse{\equal{#1}{property}}{\PropertyRefName}{}%
  \ifthenelse{\equal{#1}{question}}{\QuestionRefName}{}%
  \ifthenelse{\equal{#1}{remark}}{\RemarkRefName}{}%
  \ifthenelse{\equal{#1}{notation}}{\NotationRefName}{}%
  \ifthenelse{\equal{#1}{theorem}}{\TheoremRefName}{}%
  \ifthenelse{\equal{#1}{figure}}{\FigureRefName}{}%
  \ifthenelse{\equal{#1}{equation}}{\EquationRefName}{}%
  \ifthenelse{\equal{#1}{section}}{\SectionRefName}{}%
}

  \usepackage{atbegshi,picture}
\def\Vhrulefill{\leavevmode\leaders\hrule height 0.7ex depth \dimexpr0.4pt-0.7ex\hfill\kern0pt}
\newcommand{\runningnotice}[1]{
  \AtBeginShipout{%
    \AtBeginShipoutUpperLeft{%
      \put(\dimexpr\paperwidth-6pt\relax,-\paperheight+2.6cm){%
        \rotatebox{90}{\makebox[\paperheight-3.5cm][r]{%
          \Vhrulefill\normalfont\ttfamily\normalsize~#1~\raisebox{0.7ex}{\rule{1cm}{0.4pt}}%
        }}%
      }%
    }
  }
}

\def\jscompatibility{0}
\def\curlang{en}

\newcommandx{\vmnewcommandx}[5][2=0,3={},5={},usedefault]{
      \ifthenelse{\equal{\jscompatibility}{0}}
      {\newcommandx{#1}[#2][#3]{#4}} 
      {\newcommandx{#1}[#2][#3]{#5}} 
}

\vmnewcommandx{\wlen}[1]{|#1|}

\vmnewcommandx{\cod}[1]{\langle #1 \rangle}
\vmnewcommandx{\floor}[1]{\lfloor #1 \rfloor}
\vmnewcommandx{\ceil}[1]{\lceil #1 \rceil}

\newcommandx{\newcommandy}[5][1=i,3=0,4={}]{%
  \ifthenelse{\isundefined{#2}}{\newcommandx{#2}[#3][#4]{#5}}{%
      \ifthenelse{\equal{#1}{i}}{}{}%
      \ifthenelse{\equal{#1}{o}}{\renewcommandx{#2}[#3][#4]{#5}}{}%
    }%
}
\newcommandy[o]{\pathx}[2][2={}]{
  \nlb%
  \ifthenelse{\equal{#2}{}}%
    {\xrightarrow{\ #1 \ }}%
    {\underset{#2}{\xrightarrow{\ #1 \ }}}%
  \nlb%
}
\newcommandy[o]{\pathy}[2][2={}]{\nlb\underset{#2}{\xleftarrow{\ #1 \ }}\nlb}

\newcommand{\val}[1]{\widebar{#1}}

\newcommand{\strong}[1]{\textbf{#1}}
\newcommand{\ssc}[1]{\textbf{\textsc{#1}}}

\newcommand{\set}[1]{\{#1\}}

\newcommand{\N}{\mathbb{N}}
\newcommand{\Q}{\mathbb{Q}}

\newcommandy[o]{\Cup}{\bigcup}
\newcommand{\nlb}{\nolinebreak}

\renewcommand{\thmBlockFont}[1]{\ssc{#1}}

\newcommand{\frcurrentcaption}{}
\newcommand{\encurrentcaption}{}
\newcommandy{\vmfigure}[2]{
  \begin{figure}[ht!]
    \centering
    \input{#1/#2}
    \ifthenelse{\equal{\curlang}{fr}}
    {\caption{\frcurrentcaption}}{}
    \ifthenelse{\equal{\curlang}{en}}
    {\caption{\encurrentcaption}}{}%
    \def\vmlastfigure{#2}
    \lfigure{#2}
  \end{figure}
}

\newcommand{\vmStartTrickI}[3]{#1{#3}#2}
\newcommand{\vmStarTrickII}[4]{#1{#3}{#4}#2}

\makeatletter
\newcommand*{\addStartTextModeZ}[3][i]{ 
  \newcommandy[#1]{#2}{\protect\@ifstar{\leavevmode\protect\nlb$\protect#2$}{\protect#3}}
}
\newcommand*{\addStartTextModeI}[2]{
  \newcommand{#1}{\@ifstar{\leavevmode\nlb$\vmStartTrickI{#1}{$}}{#2}}
}
\newcommand*{\addStartTextModeII}[2]{
  \newcommand{#1}{\@ifstar{\leavevmode\nlb$\vmStarTrickII{#1}{$}}{#2}}
}
\newcommand*{\addStarTextModeIII}[2]{
  \newcommand{#1}{\@ifstar{\leavevmode\nlb$\vmStarTrickIII{#1}{$}}{#2}}
} 

\newcommand*{\addMagicMathModeZ}[3][i]{
  \newcommandy[#1]{#2}{{\ifmmode #3 \else \leavevmode\protect\nlb$\protect#3$\fi}}
}

\makeatletter

\renewcommand{\leq}{\leqslant}
\renewcommand{\geq}{\geqslant}
\renewcommand{\phi}{\varphi}
\renewcommand{\epsilon}{\varepsilon}
\renewcommand{\mod}{\text{~mod~}}

\newcommandx{\base}[2][1=p,2=q,usedefault]{\frac{#1}{#2}}
\newcommandx{\values}[1][1=\base]{V_{#1}}
\newcommandx{\natlang}[1][1=\base]{L_{#1}}
\newcommandx{\nattree}[1][1=\base]{T_{#1}}
\newcommandx{\addconst}[2][1=\base]{\mathcal{R}_{#1,#2}}

\newcommandx{\converter}[2][1=\base]{\mathcal{C}_{#1,#2}}
\renewcommand{\val}[1]{\pi\hspace*{-0.1em}\left(#1\right)}

\newcommand{\Pref}[1]{\text{Pref}\xmd(#1)}

\newcommand{\Apq}[1][p]{A_{#1}}
\newcommand{\Apqe}[1][p]{{A_{#1}}^{\!\!*}}
\newcommand{\pqRep}[1]{\langle #1 \rangle_{\frac{p}{q}}}
\newcommand{\THFracs}[2]{\frac{#1}{#2}}%
\newcommand{\pq}{\THFrac{p}{q}}
\newcommand{\pqs}{\THFracs{p}{q}}

\newcommand{\Indpq}[1]{#1_{\pqs}}%

\newcommand{\tds}{\THFracs{3}{2}}

\newcommand{\Lpq}{\Indpq{L}}

\newcommand{\Vpq}{\Indpq{V}}

\newcommandx{\yaHelper}[2][1=\empty]{%
\ifthenelse{\equal{#1}{\empty}}%
  { \ensuremath{\scriptstyle{#2}}} 
  { \raisebox{ #1 }[0pt][0pt]{\ensuremath{\scriptstyle{#2}}}}  
}

\newcommandx{\yrightarrow}[4][1=\empty, 2=\empty, 4=\empty, usedefault=@]{%
  \ifthenelse{\equal{#2}{\empty}}
  { \xrightarrow{ \protect{ \yaHelper[#4]{#3} } } } 
  { \xrightarrow[ \protect{ \yaHelper[#2]{#1} } ]{ \protect{ \yaHelper[#4]{#3} } } } 
}

\newcommandy[o]{\pathx}[2][2={}]{
  \nlb%
  \ifthenelse{\equal{#2}{}}%
    {\yrightarrow{\ #1\ }[-1pt]}%
    {\underset{#2}{\xrightarrow{\ #1 \ }}}%
  \nlb%
}

\newcommand{\path}{}
\newcommand{\fa}{\forall}

\newcommand{\e}{\text{\quad}}                 
\newcommand{\ee}{\text{\qquad}}               
\newcommand{\eee}{\text{\qquad \qquad}} 
\newsavebox{\InterSymbolSpace}
\savebox{\InterSymbolSpace}{\hspace{0.125em}}
\newsavebox{\SideFormulaSpace}
\savebox{\SideFormulaSpace}{\hspace{0.2em}}
\newcommand{\msp}{\usebox{\SideFormulaSpace}} 
\newcommand{\xmd}{\usebox{\InterSymbolSpace}} 
\newcommand{\eqpnt}{\makebox[0pt][l]{\: .}}

\newcommand{\eqvrg}{\makebox[0pt][l]{\: ,}}
\newcommand{\EqVrgInt}{\: , \e }

\newcommand{\quantvrg}{\, , \;}
\newcommand{\quantsp}{\ee }

\newcommand{\LatinLocution}[1]{{\itshape #1}\xspace}

\newcommand{\cf}{\LatinLocution{cf.}}

\newcommand{\ie}{{that is, }}

\newcommand{\UNmbb}{{\mathchoice
{\hbox{$\textstyle\rm 1\kern-0.2em I$}}%
{\hbox{$\textstyle\rm 1\kern-0.2em I$}}%
{\hbox{$\scriptstyle\rm 1\kern-0.15em I$}}%
{\hbox{$\scriptscriptstyle\rm 1\kern-0.1em I$}}%
}}

\newcommand{\Tc}{\mathcal{T}}

\newlength{\ArrowDiagSize}
\newlength{\ArrowDiagWidth}
\newpsstyle{SLDiagStyle}%
   {colsep=6ex,rowsep=5ex,nodesep=1ex,npos=.45,%
    arrows=->,linewidth=\ArrowDiagWidth,arrowsize=\ArrowDiagSize,%
        linestyle=solid,linecolor=\ArrowDiagColor}

\newenvironment{SLDiag}%
   {\psset{style=SLDiagStyle}\begin{psmatrix}}%
   {\end{psmatrix}}%
\newcommand{\CDSL}{\begin{SLDiag}}
\newcommand{\CDSLF}{\end{SLDiag}}
\newenvironment{DiagraBig}%
{\psmatrix[colsep=7ex,rowsep=6ex,arrows=->,nodesep=1ex,npos=.45]}%
{\endpsmatrix}
\newcommand{\CDB}{\begin{DiagraBig}}
\newcommand{\CDBF}{\end{DiagraBig}}
\newenvironment{DiagraSmall}%
{\psmatrix[colsep=3ex,rowsep=3ex,arrows=->,nodesep=1ex,npos=.45]}%
{\endpsmatrix}
\newcommand{\CDS}{\begin{DiagraSmall}}
\newcommand{\CDSF}{\end{DiagraSmall}}

\newcommand{\matriceuu}[1]%
    {\begin{pmatrix} #1 \end{pmatrix}}
\newcommand{\matricedd}[4]%
    {\begin{pmatrix} #1 & #2 \\ #3 & #4 \end{pmatrix}}
\newcommand{\vecteurd}[2]%
    {\begin{pmatrix} #1 \\ #2 \end{pmatrix}}
\newcommand{\ligned}[2]%
    {\begin{pmatrix} #1 & #2 \end{pmatrix}}
\newcommand{\matricett}[9]%
    {\begin{pmatrix}  #1 & #2 & #3 \\
                      #4 & #5 & #6 \\
                      #7 & #8 & #9 \end{pmatrix}}
\newcommand{\vecteurt}[3]%
    {\begin{pmatrix} #1 \\ #2 \\ #3 \end{pmatrix}}
\newcommand{\lignet}[3]%
    {\begin{pmatrix} #1 & #2 & #3 \end{pmatrix}}

\newlength{\jsWidthCol}
\newlength{\blocinterligne}
\newlength{\blocinterligned}
\newlength{\temparraycolsep}
\newlength{\longueurbloc}
\newlength{\hauteurbloc}
\newlength{\centragebloc}
\newlength{\longueurblc}
\newlength{\hauteurblc}
\newlength{\centrageblc}
\newcommand{\blocligne}[1]%
    {\framebox[\longueurbloc]{$#1$}}
\newcommand{\blocmatrice}[1]%
    {\framebox[\longueurbloc]{\rule[\centragebloc]{0mm}{\hauteurbloc}$#1$}}
\newcommand{\blocvecteur}[1]%
    {\framebox{\rule[\centragebloc]{0mm}{\hauteurbloc}$#1$}}
\newcommand{\blcligne}[1]%
    {\framebox[\longueurblc]{$#1$}}
\newcommand{\blcmatrice}[1]%
    {\framebox[\longueurblc]{\rule[\centrageblc]{0mm}{\hauteurblc}$#1$}}
\newcommand{\blcvecteur}[1]%
    {\framebox{\rule[\centrageblc]{0mm}{\hauteurblc}$#1$}}
\newcommand{\matriceddblvs}[4]
   {\setlength{\temparraycolsep}{\arraycolsep}%
    \setlength{\arraycolsep}{1.3pt}%
        \left (%
    \begin{array}{cc}%
                #1  & \blcligne{#2} \\
            \blcvecteur{#3} & \blcmatrice{#4}
        \end{array}%
        \right )%
    \setlength{\arraycolsep}{\temparraycolsep}%
   }%
\newcommand{\vecteurdblvs}[2]%
   {\setlength{\temparraycolsep}{\arraycolsep}%
    \setlength{\arraycolsep}{1.5pt}%
        \left (%
    \begin{array}{c}%
                #1  \\
            \blcvecteur{#2}
        \end{array}%
        \right )%
    \setlength{\arraycolsep}{\temparraycolsep}%
   }%
\newcommand{\lignedblvs}[2]%
   {\setlength{\temparraycolsep}{\arraycolsep}%
    \setlength{\arraycolsep}{1.5pt}%
        \left (%
    \begin{array}{cc}%
                #1  & \blcligne{#2}
        \end{array}%
        \right )%
    \setlength{\arraycolsep}{\temparraycolsep}%
   }%
\newcommand{\matricettblvs}[9]
   {\setlength{\temparraycolsep}{\arraycolsep}%
    \setlength{\arraycolsep}{1.5pt}%
        \left (%
    \begin{array}{ccc}%
                #1  & \blcligne{#2} & #3\\
            \blcvecteur{#4} & \blcmatrice{#5} & \blcvecteur{#6}\\
                #7  & \blcligne{#8} & #9\\
        \end{array}%
        \right )%
    \setlength{\arraycolsep}{\temparraycolsep}%
   }%
\newcommand{\vecteurtblvs}[3]%
   {\setlength{\temparraycolsep}{\arraycolsep}%
    \setlength{\arraycolsep}{1.5pt}%
        \left (%
    \begin{array}{c}%
                #1  \\
            \blcvecteur{#2}\\
                #3
        \end{array}%
        \right )%
    \setlength{\arraycolsep}{\temparraycolsep}%
   }%
\newcommand{\lignetblvs}[3]%
   {\setlength{\temparraycolsep}{\arraycolsep}%
    \setlength{\arraycolsep}{1.5pt}%
        \left (%
    \begin{array}{ccc}%
                #1  & \blcligne{#2} & #3
        \end{array}%
        \right )%
    \setlength{\arraycolsep}{\temparraycolsep}%
   }%
\newcommand{\matricettblblvs}[9]
   {\setlength{\temparraycolsep}{\arraycolsep}%
    \setlength{\arraycolsep}{1.5pt}%
        \left (%
    \begin{array}{ccc}%
                #1  & \blcligne{#2} & \blcligne{#3}\\
            \blcvecteur{#4} & \blcmatrice{#5} & \blcmatrice{#6}\\
                \blcvecteur{#7}  & \blcmatrice{#8} & \blcmatrice{#9}\\
        \end{array}%
        \right )%
    \setlength{\arraycolsep}{\temparraycolsep}%
   }%
\newcommand{\vecteurtblblvs}[3]%
   {\setlength{\temparraycolsep}{\arraycolsep}%
    \setlength{\arraycolsep}{1.5pt}%
        \left (%
    \begin{array}{c}%
                #1  \\
            \blcvecteur{#2}\\
                \blcvecteur{#3}
        \end{array}%
        \right )%
    \setlength{\arraycolsep}{\temparraycolsep}%
   }%
\newcommand{\lignetblblvs}[3]%
   {\setlength{\temparraycolsep}{\arraycolsep}%
    \setlength{\arraycolsep}{1.5pt}%
        \left (%
    \begin{array}{ccc}%
                #1  & \blcligne{#2} & \blcligne{#3}
        \end{array}%
        \right )%
    \setlength{\arraycolsep}{\temparraycolsep}%
   }%
\newcommand{\PushLine}{\hbox{}\hfill\hbox{}}
\newlength{\DefiTest}
\newlength{\DefiHeightu}\newlength{\DefiHeightd}%
\newlength{\DefiDepthu}\newlength{\DefiDepthd}%
\newcommand{\Defi}[2]%
    {%
     \settoheight{\DefiHeightu}{${\displaystyle #1}$}%
     \settodepth{\DefiDepthu}{${\displaystyle #1}$}%
     \addtolength{\DefiHeightu}{\DefiDepthu}%
     \settoheight{\DefiHeightd}{${\displaystyle #2}$}%
     \settodepth{\DefiDepthd}{${\displaystyle #2}$}%
     \addtolength{\DefiHeightd}{\DefiDepthd}%
     \left\{#1%
     \rule[-\DefiDepthd]{\DefiTest}{\DefiHeightd}%
     \xmd\right|%
     \left.%
     \rule[-\DefiDepthu]{\DefiTest}{\DefiHeightu}%
      #2\right\}%
     }
\newlength{\ColoText}
\newlength{\ColoFigu}
\newlength{\TextFiguSpace}
\newlength{\parindenttemp} 
\newlength{\parskiptemp} 
\newlength{\fboxseptemp} 
\newcommand{\TFBoxing}{}
\newcommand{\TFVertAlig}{}
\newcommand{\LeftLarg}{}
\renewcommand{\LeftLarg}{.66}
\ifdraft\renewcommand{\TFBoxing}{\fbox}\fi
\newcommand{\TxtFg}[3]%
   {%
    \setlength{\ColoText}{#1\textwidth}%
    \setlength{\ColoFigu}{\textwidth}%
    \addtolength{\ColoFigu}{-\ColoText}%
    \addtolength{\ColoText}{-.5\TextFiguSpace}%
    \addtolength{\ColoFigu}{-.5\TextFiguSpace}%
    \ifdraft\setlength{\fboxsep}{0pt}\fi
    \noi
    \TFBoxing{%
       \begin{minipage}[\TFVertAlig]{\ColoText}%
          \setlength{\parindent}{\parindenttemp}%
          \setlength{\parskip}{\parskiptemp}%
          \par\vspace*{0mm}
             #2
       \end{minipage}%
             }%
    \hspace*{\TextFiguSpace}%
    \TFBoxing{%
       \begin{minipage}[\TFVertAlig]{\ColoFigu}%
          \par\vspace*{0mm}%
             #3%
       \end{minipage}%
             }%
    \ifdraft\setlength{\fboxsep}{\fboxseptemp}\fi
   }%
\newcommand{\TextFigu}[3][\LeftLarg]%
   {\renewcommand{\TFVertAlig}{t}\TxtFg{#1}{#2}{#3}}
\newcommand{\TextFiguC}[3][\LeftLarg]%
   {\renewcommand{\TFVertAlig}{c}\TxtFg{#1}{#2}{#3}}
\newcommand{\TextFiguX}[3][\LeftLarg]
   {%
    \setlength{\ColoText}{#1\textwidth}%
    \setlength{\ColoFigu}{\textwidth}%
    \addtolength{\ColoFigu}{-\ColoText}%
    \addtolength{\ColoText}{-.5\TextFiguSpace}%
    \addtolength{\ColoFigu}{-.5\TextFiguSpace}%
    \addtolength{\ColoFigu}{\ETAExtendedLineWidth}
    \ifdraft\setlength{\fboxsep}{0pt}\fi
    \noi
    \ifodd\value{page}%
       \TFBoxing{%
          \begin{minipage}[t]{\ColoText}%
             \RstBLS
             \setlength{\parindent}{\parindenttemp}%
             \setlength{\parskip}{\parskiptemp}%
             \par\vspace*{0mm}
                #2
          \end{minipage}%
                }%
       \hspace*{\TextFiguSpace}%
       \TFBoxing{%
          \begin{minipage}[t]{\ColoFigu}%
             \par\vspace*{0mm}%
                #3%
          \end{minipage}%
                }%
    \else
       \hspace*{-\ETAExtendedLineWidth}
       \TFBoxing{%
          \begin{minipage}[t]{\ColoFigu}%
             \par\vspace*{0mm}%
                #3%
          \end{minipage}%
                }%
       \hspace*{\TextFiguSpace}%
       \TFBoxing{%
          \begin{minipage}[t]{\ColoText}%
             \RstBLS
             \setlength{\parindent}{\parindenttemp}%
             \setlength{\parskip}{\parskiptemp}%
             \par\vspace*{0mm}
                #2
          \end{minipage}%
                }%
    \fi%
    \ifdraft\setlength{\fboxsep}{\fboxseptemp}\fi
   }
\newcommand{\Axio}[1]%
   {\pointn #1\hspace*{.1em}\jspointtiret\hspace*{.4em}\ignorespaces}
\newcommand{\ExtnF}[1]%
   {\overset{{\scriptscriptstyle \pmb{\smile}}}{#1}}

\newcommand{\DiffF}[1]%
   {\overset{{\scriptscriptstyle \pmb{\lor}}}{#1}}

\newcommand{\LocaF}[1]%
   {\overset{{\scriptscriptstyle \leftrightarrow}}{#1}}

\newcommand{\jsDist}[2][{}]%
   {\operatorname{\mathbf{d}_{#1}}\left(#2\right)}

\renewcommand{\lim}{{\operatornamewithlimits{\mathsf{lim}}}}

\newcommand{\SerSAnMon}[2]%
    {#1 \langle \! \langle  #2  \rangle \! \rangle }
\newcommand{\SerSAnMonD}[2]%
    {\left[#1\right] \langle \! \langle  #2  \rangle \! \rangle }
\newcommand{\SerMon}[1]%
    {\!\langle \! \langle  #1  \rangle \! \rangle }
\newcommand{\PolSAnMon}[2]%
    {{#1 \langle  #2 \rangle }}
\newcommand{\PolMon}[1]%
    {{\!\langle  #1 \rangle }}

\newsavebox{\LeftBraket}
\savebox{\LeftBraket}{\scalebox{0.7 1.2}{$<$}}
\newsavebox{\RightBraket}
\savebox{\RightBraket}{\scalebox{0.7 1.2}{$>$}}

\newcommand{\jsStar}[1]{{{#1}^{*}}}
\newcommand{\Ae}{\jsStar{A}}

\newcommand{\iotaK}{\iota_{\ShiftInd{K}}}

\newcommand{\compos}{\ccdot }

\newcommand{\phiikpsi}%
{{\varphi ^{-1}\! \compos        \iotaK \! \compos \! \psi }}
\newcommand{\phiiotpsi}[1]%
{{\varphi ^{-1}\! \compos        \iota _{\ShiftInd{#1}} \! \compos \! \psi }}
\newcommand{\phiintkpsi}[1]%
{{(#1\varphi ^{-1}\! \cap K) \psi }}

\newcommand{\jsless}
   {\mathrel{\leqslant_{\!\!\!\!\scriptscriptstyle{/}}}}
\newcommand{\jsgrea}
   {\mathrel{\geqslant_{\!\!\!\!\scriptscriptstyle{\backslash}}}}

\newcommand{\lexiconeq}
   {\preccurlyeq_{\!\!\!\!\!\scalebox{1.8 1}{\scriptscriptstyle{\pmb{/}}}}}

\newcommand{\jsAutUn}[1]%
   {\mbox{$\left\langle \thinspace #1 \thinspace \right\rangle $}}
\newcommand{\aut}[1]{\jsAutUn{#1}} 

\newcommand{\ShiftInd}[1]{\raisebox{-0.3ex}{$\scriptstyle{#1}$}}

\newcommand{\actb}{\mathbin{\raisebox{0.2ex}%
                        {${\scriptscriptstyle \circ} $}}}
\newcommand{\ccdot}{\actb} 

\newlength{\vbh}\newlength{\vbd}\newlength{\vbt}%
\newcommand{\CompAuto}[1]%
    {%
     \settodepth{\vbd}{\mbox{$\displaystyle{#1\strut}$}}%
     \settoheight{\vbh}{\mbox{$\displaystyle{#1\strut}$}}%
     \setlength{\vbt}{\vbh}\addtolength{\vbt}{\vbd}%
     {}%
     \psline[linewidth=0.8pt]{c-c}(0,-.65\vbd)(0,.9\vbh)%
     \hspace*{0.7pt}%
     {#1}%
     \kern0.8pt%
     \psline[linewidth=0.8pt]{c-c}(0,-.65\vbd)(0,.9\vbh)%
     }%

\newcommand{\bornedeuxlignes}[2]%
{\mbox{$
\begin{array}{c}{\scriptstyle #1}\\ {\scriptstyle #2} \end{array}
       $}}

\renewcommand{\path}[1]{\xrightarrow{\ #1 \ }} 
\newcommand{\pathaut}[2]{\underset{#2}{\path{#1}}}

\newcommand{\ExpDer}[2][a]%
    {\operatorname{\frac{\partial}{\partial \mbox{$#1$}}}#2}
\newcommand{\ExpDerP}[2][a]%
    {\operatorname{\frac{\partial}{\partial\mbox{$#1$}}}\left(#2\right)}
\newcommand{\ExpDerr}[2][a]%
    {\operatorname{\frac{\partial_{\mathrm{R}}}{\partial \mbox{$#1$}}}#2}
\newcommand{\ExpDerB}[2][a]%
   {\operatorname{\frac{\partial_\mathsf{b}}{\partial \mbox{$#1$}}}#2}
\newcommand{\ExpDerBP}[2][a]%
   {\operatorname{\frac{\partial_\mathsf{b}}{\partial \mbox{$#1$}}}\left(#2\right)}

\renewcommand{\pq}{\base}

\newcommand{\defin}[1]{Definition~\ref{d.#1}}

\newcommand{\figur}[1]{Figure~\ref{f.#1}}
\newcommand{\lemme}[1]{Lemma~\ref{l.#1}}
\newcommand{\propo}[1]{Proposition~\ref{p.#1}}

\newcommand{\theor}[1]{Theorem~\ref{t.#1}}
\nolinenumbers
\renewcommand{\runningnotice}[1]{
  \AtBeginShipout{%
    \AtBeginShipoutUpperLeft{%
      \put(\dimexpr\paperwidth-16pt\relax,-\paperheight+1.5cm){%
        \rotatebox{90}{\makebox[\paperheight-3cm][r]{%
          \Vhrulefill\normalfont\ttfamily\normalsize~#1~\raisebox{0.7ex}{\rule{1cm}{0.4pt}}%
        }}%
      }%
    }
  }
}
\ifdraft
  \linenumbers
  \usepackage[notref,notcite]{showkeys}
  
  \runningnotice{RatBaseFinGem v4 $\bullet$ Victor Marsault $\bullet$ Jacques Sakarovisizetch $\bullet$ \today { }$\bullet$ \currenttime}
  \usepackage[backgroundcolor=white,bordercolor=orange]{todonotes}
  \ShowFrame
\else
  \usepackage[disable]{todonotes}
\fi

\def\curlang{en}
\newcommand{\langdep}[2]{\ifthenelse{\equal{\curlang}{fr}}{#1}{#2}}
\title {On sets of numbers rationally represented\\
in a rational base number system%
.}

\date{\today}
\author
{
        Victor Marsault\thanks{Corresponding author, {\tt victor.marsault@telecom-paristech.fr}}${}^{~,}$\footnotemark 
        \addtocounter{footnote}{-1}
        \and
        Jacques Sakarovitch\thanks{Telecom-ParisTech and CNRS, 46 rue Barrault 75013 Paris, France}
}

\makeatletter
\let\newtitle\@title
\let\newauthor\@author
\let\newdate\@date
\makeatother
\begin{document}
\pretolerance=4000
\tolerance=5000
\maketitle
\thispagestyle{plain}

\begin{abstract}
In this work, it is proved that a set of numbers closed under addition and whose
  representations in a rational base numeration system is a rational
  language is not a finitely generated additive monoid.

A key to the proof is the definition of a strong combinatorial property on languages : 
  the bounded left iteration property.
It is both an unnatural property in usual formal language theory (as it 
  contradicts any kind of pumping lemma) and an ideal fit to the languages 
  defined through rational base number systems.
\end{abstract}

\section{Introduction}
The numeration systems in which the base is a rational number have been
introduced and studied in~\cite{AkiyEtAl08}.
It appeared there that the language of representations of all integers in such a 
system is ``complicated'', by reference to the classical Chomsky 
hierarchy and its usual iteration properties.
This work is a contribution to a better understanding of the structure of this
language.
It consists in a result whose statement first requires some basic
facts about number systems.

Given an integer $p$ as a base, the set of non-negative integers~$\N$ is
represented by the set of words on the alphabet
$\Apq=\set{0,1,\ldots,(p-1)}$ which do not begin with a~$0$.
This
set~$L_p=(\Apq\backslash 0)\Apqe$ is rational, \ie accepted by a
finite automaton.  This representation of integers has another
property related to finite automata: the addition is realised by a
finite 3-tape automaton.

This addition algorithm can be broken down into two steps : first a
digit-wise addition which outputs a word on the double alphabet
$A_{2p-1}$ whose value in base $p$ is the sum of the two input words;
second a transformation of a word of ${(A_{2p-1})}^*$ into a word of
$\Apqe$ without modifying its value.  This second step can be done by a finite
transducer called the \emph{converter} (see~Section~2.2.2 of~\cite{FrouSaka10hb}).

Many non-standard numeration systems that have been studied so far
have the property that the set of representations of the integers is
a rational language.
It is even \emph{the} property that is retained in the study of the
\emph{abstract numeration systems}, even if
it is not the case that addition can be realised by a finite
automaton (\cf~\cite{LecoRigo10hb}).

In the rational base numeration systems, as defined and studied
in~\cite{AkiyEtAl08}, the situation is reverse: the set of integers
is not represented by a rational language (not even a context-free
one), but nevertheless the addition is realised by a finite automaton.
More precisely, let~$p$ and~$q$ be two coprime integers, with~${p>q}$.
In the $\pq$-numeration system, the digit alphabet is again~$\Apq$, and the
value of a word~$\msp{u = a_n\cdots\xmd a_2\xmd a_1}\msp$ in~$\Apqe$
is~$\val{u}=\frac{1}{q}\sum_{i=0}^n a_{i}(\frac{p}{q})^i$.
In this system, every integer has a unique finite representation, but the
set~$\Lpq$ of the $\pq$-representations of the integers is not a
rational language.
The set~$\Vpq$ of all numbers that can be
represented in this system, $\Vpq=\val{\Apqe}$, is closed under
addition but is not finitely generated (as an additive monoid).



In this work, we establish the contradiction between being a finitely
generated additive monoid and having a rational set of
representations in a rational base number system.

\begin{theorem}
\label{t.fin-gen-mon}%
  The set of the $\base$-representations of any finitely generated
  additive submonoid of~$\Vpq$ is not a rational language.
\end{theorem}

The proof of this statement relies on three ingredients.
The first one is the description of a weak iteration property
whose negation is satisfied by the language~$\Lpq$.
The second one is the construction of a sequential letter-to-letter right
 transducer that realises, on the~$\pq$-representations,
the addition of a fixed value to the elements of~$\Vpq$.
Finally, the third one is a characterisation of a finitely generated 
additive submonoid of~$\Vpq$ as a finite union of translates of the 
set of the integers.

The paper is organised as follows: after the preliminaries, where we essentially
recall the definition of transducers, we present with more details in 
\rsection{rat-base} the numeration system in base~$\pq$ .
In \rsection{blip}, we describe the Bounded Left Iteration Property (BLIP) and in
\rsection{incrementer}, we build a transducer called incrementer.
In the last section, we give the proof of a much stronger statement 
than~\rtheorem{fin-gen-mon}, expressed with the BLIP property.


\section{Preliminaries}\lsection{prelim}

We essentially follow notations and definitions of \cite{Saka09} for automata and
  transducers.
An \emph{alphabet} is a finite set of \emph{letters},
  the \emph{free monoid} generated by~$A$, and denoted by~$A^*$, 
  is the set of finite
  \emph{words} over~$A$.
The \emph{concatenation} of two words~$u$ 
  and~$v$ of~$A^*$ is denoted by~$u\xmd v$,
  or by~$\msp u.v\msp$ when the dot adds hopefully to readability.
A \emph{language} (over~$A$) is any subset of~$A^*$.

\smallskip

A language is said to be \emph{rational} (resp. \emph{context-free}) if it is 
  accepted by a finite automaton (resp. a \text{pushdown automaton}). 
The precise definitions of these classes of automata are however irrelevant 
  to the present work, and can be found in~\cite{HopcMotwUllm00}.
Similarly, we are only considering (and thus defining) a very restricted 
  class of transducers, namely the \emph{sequential letter-to-letter} 
  transducer.

%

\medskip

Given two alphabets~$A$ and~$B$, 
a sequential letter-to-letter (left) transducer~$\Tc$ from~$A^*$ to~$B^*$ is a directed graph whose
edges are labelled in~$A\times B$.
More precisely,~$\Tc$ is defined by a 6-tuple~$\msp\Tc=\aut{Q,A,B,\delta,\eta,i,\omega}\msp$
where Q is the set of \emph{states};~$A$ is the \emph{input alphabet};~$B$ is the 
\emph{output alphabet};~$\delta:Q\times A\rightarrow Q$ is the 
\emph{transition function};~$\eta:Q\times A\rightarrow B$ is the \emph{output function};~$i$ 
is the \emph{initial state} and~$\omega:Q\rightarrow B^*$
is the \emph{final function}. 

Moreover, we call \emph{final} any state in the definition domain
of~$\omega$. As usual, the function~$\delta$ (resp.~$\eta$) is extended 
  to~$Q\times A^* \rightarrow Q$ (resp. $Q\times A^* \rightarrow B^*$) 
  by~${\delta(p,\epsilon)=p}$ (resp.~$\eta(p,\epsilon)=\epsilon$)
  and~$\delta(p,a.u)=\delta(\delta(p,a),u)$
  (resp.~$\eta(p,a.u)=\eta(p,a).\eta(\delta(p,a),u)$).
  
Given~$\Tc$, we write~$\msp{p\pathx{u~|~v}[\Tc]q}\msp$
  if, and only if,~${\delta(p,u)=q}$ and~$\eta(p,u)=v$.
By analogy, we denote by~$\msp p\pathx{w}[\Tc]\msp$ the fact that~$p$ is a final state
  and that~$\omega(p)=w$.
%
%
The \emph{image} by~$\Tc$ of a word~$u$, denoted by~$\Tc(u)$, 
  is the word~$\msp v.w\msp $,
  if~$\msp{i\pathx{u~|~v}[\Tc]p\pathx{w}[\Tc]}\msp$.


\smallskip

Finally, a transducer is said to be a \emph{right transducer}, if it reads the 
words from right to left; and to be \emph{complete} if both the transition function
and the output function are total functions.

\medskip

In the following, every considered transducer will be complete, letter-to-letter,
right and sequential.

\section{Rational base number system}\lsection{rat-base}
We recall here the definitions, notations and 
constructions of~\cite{AkiyEtAl08}.
Let~$p$ and~$q$ be two coprime integers such that~$p>q>1$.
Given a positive integer~$N$, let us define~$N_0=N$ and for all~$i>0$:
\begin{equation}
  q\xmd N_i = pN_{i+1} + a_i
\end{equation}
where~$a_i$ is the remainder of the Euclidean division of $q\xmd N_i$ by~$p$,
hence in~$A_p$.
Since~$p>q$, the sequence~$(N_i)_i$ is strictly decreasing and eventually
stops at $N_{k+1}=0$. Moreover the equation
\begin{equation}
  N = \sum^{k}_{i=0} \frac{a_i}{q} \left( \pq \right)^i
\end{equation}

\noindent holds. The evaluation function~$\pi$ is derived from this formula. 
The value of a word~${u=a_na_{n-1}\cdots\xmd a_0}$ over~$A_p$ is defined as
\begin{equation}
  \val{a_na_{n-1}\cdots\xmd a_0} = \sum^{n}_{i=0} \frac{a_i}{q} \left( \pq \right)^i
\end{equation}

Conversely, a word~$u$ is called a~$\pq$-representation of
a number~$x$ \linebreak if~${\val{u}=x}$.
Since the representation is unique up to leading 0's
(see~{\cite[\nlb Theorem~1]{AkiyEtAl08}}),~$u$ is denoted by~$\cod{x}_{\pq}$
(or~$\cod{x}$ for short),
and in the case of integers, can be computed with the modified
Euclidean division algorithm above.
By convention, the representation of~0 is the empty word~$\epsilon$.

\medskip

It should be noted that a rational base number systems is \strong{not} 
a~$\beta$-numeration (\cf \cite[Chapter~7]{Loth02}) in the special case
where~$\beta$ is rational.
In the latter, the digit set is~$\set{0,1,\ldots,\ceil{\pq}}$
and the weight of the~$i$-th leftmost
digit is~$(\pq)^i$; whereas in rational base number systems, they respectively
are~$\set{0,1,\ldots,(p-1)}$ and~$\frac{1}{q}(\pq)^i$.

\begin{definition}
  The representations of integers in the~$\pq$-system form a language 
  over~$A_p$, which is denoted
  by~$\Lpq$.
\end{definition}

It is immediate that~$\Lpq$ is prefix-closed (since, in the modified Euclidean
division algorithm~$\cod{N}=\cod{N_1}.a_0$) and prolongable
(there exists an~$a$ such that~$q$ divides~$(np+a)$ and then~${\cod{\frac{np+a}{q}}=\cod{n}.a}$).
As a consequence,~$\Lpq$ can be represented as a tree whose branches are all
  infinite (\cf\rfigure{l32}).
\begin{figure}
  \centering
  \newlength{\vmord}\setlength{\vmord}{1.25cm}
  \newlength{\vmabs}\setlength{\vmabs}{2.5cm}

  \BoxedEdgeLabelOn
  \renewcommand{\BxdEdgLblSep}{1}
  \renewcommand{\BxdEdgLblStl}{none}
  \scalebox{0.5}{\input{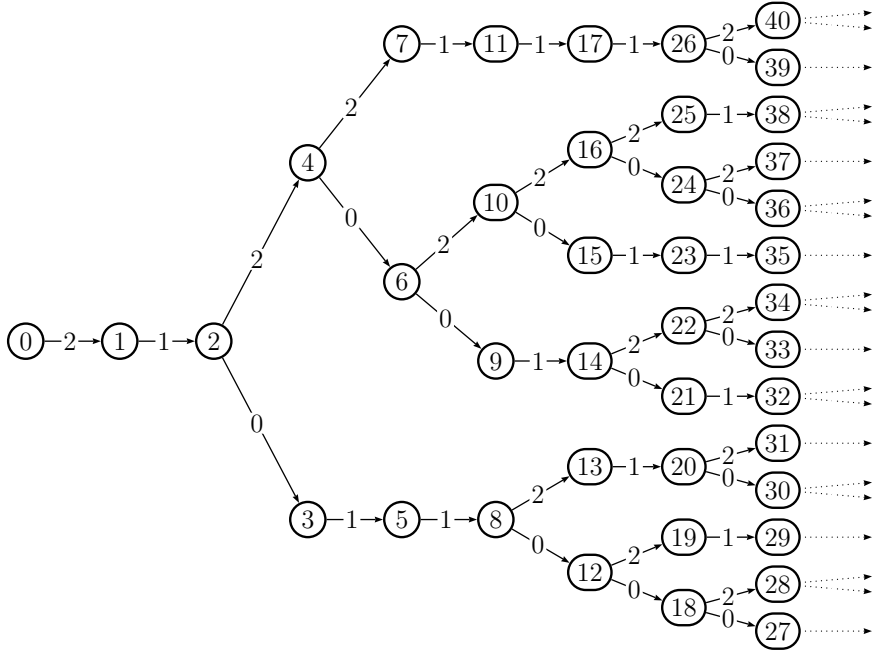}}
  \caption{The tree representation of the language~$L_{\frac{3}{2}}$}
  \lfigure{l32}
\end{figure}
On the other hand, the suffix language of~$\Lpq$ is all~$A_p^*$, and, moreover, 
every suffix appears periodically as established by the following:

\begin{proposition}[{\cite[Proposition 10]{AkiyEtAl08}}]
  For every word~$u$ over~$A_p$ of length~$k$, there exists an
  integer~$n<p^k$ such that~$u$ is a suffix of~$\cod{m}$ if, and only if,~$m$ is
  congruent to~$n$ modulo~$p^k$.
\end{proposition}

In short, the congruence modulo~$p^k$ of~$n$ determines the suffix of length~$k$
of~$\cod{n}$. In contrast, the congruence modulo~$q^k$ of~$n$ determines the
words of length~$k$ appendable to~$\cod{n}$ in order to stay in~$\Lpq$, as is
stated in the next lemma.

\begin{lemma}[{\cite[Lemma 6]{AkiyEtAl08}}]\llemma{future}
  Given two integers~$n,m$ and a word~$u$ over~$A_p$:
  \begin{enumerate}[label=(\roman{*})]
    \item if both~$\cod{n}.u$ and~$\cod{m}.u$ are in~$\Lpq$, then~$n \equiv m~[q^{\wlen{u}}]$
    \item if $n \equiv m~[q^{\wlen{u}}]$, $\cod{n}.u$ is in~$\Lpq$ implies~$\cod{m}.u$~is in~$\Lpq$.
  \end{enumerate}
\end{lemma}
\begin{proof}

  (i). The word~$\cod{n}.u$ is in~$\Lpq$ if, and only,
  if~$(n(\pq)^{\wlen{u}}+\val{u})$ is an integer, and similarly for~$m$.
  It follows that $(n-m)(\pq)^{\wlen{u}}$ is equal to some integer~$z$,
  and then~$(p^{\wlen{u}})(n-m)=zq^{\wlen{u}}$, hence~$n \equiv m~[q^{\wlen{u}}]$.

  (ii). Analogous to (i).
\end{proof}

A direct consequence of this lemma is that given any two distinct words~$u$ 
and~$v$ of~$\Lpq$, there exists a word~$w$ such that~$uw$ is in~$\Lpq$ 
but~$vw$ is not. 
Hence, the set~$\{ u^{-1}\Lpq~|~ u\in A_p^{\xmd *} \}$ of left quotients of~$\Lpq$
is infinite, or equivalently:
\begin{corollary}
  The language~$\Lpq$ is not rational.
\end{corollary}

\begin{definition}[The value set]\ldefinition{vpq}
  We denote by~$\Vpq$ the set of numbers representable in base~$\pq$, namely:
  \begin{equation}
    \Vpq = \set{x~|~\exists u\in A_p^*, \val{u}=x}
  \end{equation}
  or equivalently~$\Vpq=\val{A_p^*}$
\end{definition}

The most notable property of~$\Vpq$ is that it is closed under addition, or
more precisely that the addition is realised by a transducer, described in 
\rsection{incrementer} (a full proof can be found in 
\cite[Section 3.3]{AkiyEtAl08}).

Secondly, from the definition of~$\pi$, one derives easily
that~$\Vpq\subseteq\Q$. More precisely~$\Vpq$ contains only
numbers of the form~$\frac{x}{y}$ where y divides a power of~$q$,
and conversely, for all~$k$,~$\Vpq$ contains almost every number~$\frac{x}{q^k}$.

\begin{lemma}
\label{l.thr-esh}
For every integer~$k$, there exits an integer~$m_k$ such that, for
every integer~$n$ greater than~$m_k$, $\frac{n}{q^k}$ belongs to~$\Vpq$.
\end{lemma}

\begin{proof}
If~$k = 0$, then one can take~$m=0$ since~$\N$ is contained in~$\Vpq$.

For $k \geq 1$, the words~$1$ and~$1.0^{(k-1)}$ have for respective
value~$\frac{1}{q}$ and~$\frac{p^{k-1}}{q^{k}}$.
For every integer~$i$ and~$j$, the number
$(\frac{i\times p^{(k-1)}+j\times q^{(k-1)}}{q^{k}})$ is in $\Vpq$,
since $\Vpq$ is closed under addition,
and this can be rewritten 
as~$(p^{(k-1)}\N+q^{(k-1)}\N)\frac{1}{q^{k}} \subseteq \Vpq$.
Since~$p^{(k-1)}$ and $q^{(k-1)}$ are coprime,
($p^{(k-1)}\N+q^{(k-1)}\N$) ultimately covers $\N$.
\end{proof}

Experimentally, the bound~$m_k$ is increasing with~$k$ but the expression
resulting from this Lemma is far from being tight.
As a consequence, it proves to be difficult to define~$\Vpq$ without using
the~$\pq$-rational base number system.

\section{BLIP languages}\lsection{blip}

In the previous section, an insight is given about why~$\Lpq$ is not rational.
It is additionally proven in~\cite{AkiyEtAl08} that~$\Lpq$ is not
context-free either.
However, being context sensitive doesn't seem to accurately describe~$\Lpq$.
This section depicts a very strong language property, taylored to
capture the structural complexity of~$\Lpq$.

Let us first define a (very) weak iteration property for languages:

\begin{definition}
A language~$L$ of~$\Ae$ is said to be \emph{left-iterable} if there
exist two words $u$ and~$v$ in~$\Ae$ such that~$u\xmd v^{i}$ is a
prefix of words in~$L$ for an infinite number of exponents~$i$.
\end{definition}

Of course, every rational or context-free language is left-iterable.
The definition is indeed designed above all for stating its negation.

\begin{definition}\ldefinition{blip}
A language~$L$ which is not {left-iterable} is said to have \emph{the
Bounded Left-Iteration Property}, or, for short, to be \emph{BLIP}.
\end{definition}

\begin{example}\lexample{blip}
  A very simple way of building BLIP languages is to consider 
  infinitely many prefixes of an infinite  and aperiodic word. 
  For instance the language~$\set{u_i}$, where~$u_0=\epsilon$ and~$u_{i+1}=u_{i}.1.0^i$;
  or the language of the finite powers of the Fibonacci 
  morphism~$\set{\sigma^i(0)}$ where~${\sigma(0)=01}$ and~${\sigma(1)=0}$.

  In order to build a less trivial example let us define the following family of
  functions~$f_i$:\\
  $$\begin{array}{lll}
    f_i~:~&n \mapsto n& \text{if }n\neq i \\
         & n \mapsto 0& \text{if }n = i. \\
  \end{array}$$
  The language~$\set{u_{i,j}}$, where~$u_{i,0}=1$ and~$u_{i,{j+1}} = u_{i,j}.1.0^{f_i(j)}$,
    is BLIP as can be easily checked.
\end{example}

Since \rdefinition{blip} was taylored for the study of~$\Lpq$, the
following holds, as essentially established in \cite[Lemma~8]{AkiyEtAl08}.

\begin{proposition}
\label{l.nat-lan-bli}
The language $\Lpq$ is BLIP.
\end{proposition}
\begin{proof}
  If~$\Lpq$ were left iterable, there would exist two nonempty
  words~$u$ and~$v$ such that~$u\xmd v^i$ is prefix
    of a word of~$\Lpq$ for infinitely many~$i$.
    Since~$\Lpq$ is prefix-closed, the word~$u\xmd v^i$ would be itself in~$\Lpq$,
    for all~$i$.
    From \rlemma{future}, it follows that the integers~$\val{u}$ and~$\val{uv}$
      are congruent modulo~$q^k$, for all~$k$, a contradiction.
\end{proof}

Being BLIP is a very stable property for languages, as expressed by
the following properties.

\begin{lemma}\llemma{blip}
  \begin{minipage}[t]{0.8\linewidth}
    \begin{enumerate}[label=(\roman{*})]
      \item {Every finite language is BLIP.}
      \item {Any finite union of BLIP languages is BLIP.}
      \item {Any intersection of BLIP languages is BLIP.}
      \item {Any sublanguage of a BLIP language is BLIP.}
    \end{enumerate}
  \end{minipage}
\end{lemma}

Of course, BLIP languages are not closed under complementation, star or
transposition.

The bounded left iteration property can be expressed 
with the more classical notion of IRS language (for Infinite Regular Subset)
that has been introduced by Sheila Greibach in her study of the family of 
context-free languages (\cite{Grei1975a}, \cf also~\cite{AuteEtAl1982a}).
A language is IRS if it does not contain any infinite rational sublanguage.
For instance, the language~$\set{a^n~|~n\text{ is a prime number}}$ is IRS
(but not BLIP).

It is immediate that a BLIP language is IRS; even that a BLIP language contains
no  infinite context-free sublanguage.
However the converse is not true as seen with the above example.
More precisely, the following statement holds:

\begin{proposition}\lproposition{blip=p-irs}
  A language~$L$ is BLIP if, and only if,~$\Pref{L}$ is IRS.
\end{proposition}
\begin{proof}~\\

  \vspace*{-3em}\begin{align*}
    \Pref{L}\text{ is not IRS}
    &\iff \Pref{L}\text{ contains a sublanguage of the form}~u\xmd v^*w\\
    &\iff uv^*\text{ is a sublanguage of }\Pref{L} \\
    &\iff \text{for infinitely many }i,~u\xmd v^i\text{ is prefix of a word of }L\\
    &\iff L\text{ is not BLIP}
  \end{align*}
\end{proof}

\rproposition{blip=p-irs} shows that BLIP and IRS are equivalent properties on 
prefix-closed languages, which means that IRS is indeed a very strong property
for prefix-closed languages.

\medskip

Even though the purpose of this work is to prove~\theor{fin-gen-mon},
we actually prove a stronger version of it:
\begin{theorem}\ltheorem{blip-fin-gen}
 The set of the $\base$-representations of any finitely generated
  additive submonoid of~$\Vpq$ is a BLIP language.
\end{theorem}

This is not a minor improvement, as it shows that every language
representing a finitely generated monoid is basically as complex as~$\Lpq$.

\section{The incrementer}\lsection{incrementer}

The purpose of this section is to build a letter-to-letter sequential 
right transducer~$A_p \rightarrow A_p$
realising a constant addition:
given as \emph{parameter} a word~$w$ of~$A_p^*$ it would perform the
application~$u\mapsto v$, such that~$\val{v}=\val{u}+\val{w}$.
This transducer is based on the converter defined in \cite{FrouSaka10hb}
that we recall in \rdefinition{con-ver}, below.


\begin{theorem}[\cite{AkiyEtAl08},\cite{FrouSaka10hb}]
Given any digit alphabet~$A_n$, there exists a finite letter-to-letter
right sequential transducer~$\converter{n}$ from~$A_n$ to~$\Apq$ such
that for every~$w$
in ${A_n}^*$, $\val{\converter{n}(w)}=\val{w}$.
\end{theorem}

\begin{definition}
	\label{d.con-ver}%
For every integer $n$, the \emph{converter}
$\converter{n} = \aut{\N, A_n, \Apq, 0, \delta,\eta, \omega}$,
is the right transducer
with input alphabet~$A_{n}$, output
alphabet~$\Apq$, and whose
transition and output functions are defined by:
\begin{equation}
\forall s\in\N\quantvrg\forall a\in A_n\quantsp
s \pathaut{a|c}{} s' \iff q\xmd s + a = p\xmd s' + c
\eqvrg
\eee\ee
\notag
\end{equation}
and
final function by:
$\omega(s) = \pqRep{s}$,
for every state~$s$ in~$\N$.
\end{definition}

\defin{con-ver} describes a transducer with
an infinite number of
states, but its reachable part is finite
(cf~\cite[Proposition~13]{AkiyEtAl08} or \cite[Section~2.2.2]{FrouSaka10hb}).
In particular, if~${n = 2p-1}$, the converter is in fact an
additioner: given two words $u=a_n\cdots\xmd a_2\xmd a_1$
and $v=b_n\cdots\xmd b_2\xmd b_1$ over $\Apq$,
the digit-wise addition yields the word~$(a_n+b_n)\cdots(a_1+b_{1})$
over~$A_{2p-1}$ which is transformed by~$\converter{2p-1}$
into~$\pqRep{\val{u}+\val{v}}$. The converter from~$A_{5}$ to $A_{3}$ in base 
$\tds$ is shown at \figur{Z53}.%

\begin{figure}[ht]
  \centering
  \input{converter}
  \caption{The converter $\converter[\frac{3}{2}]{5}$}
  \label{f.Z53}
\end{figure}

For every word $w$ of $\Apqe$, we define a letter-to-letter sequential 
right transducer $\addconst{w}$
which increments the input by $w$, \ie given a word $u$ as input, it
outputs the $\base$-representation $\pqRep{\val{u}+\val{w}}$.
It is obtained as a specialisation of~$\converter{2p-1}$.

\begin{definition}
	\label{d.inc-rem}%
For every $w=b_{n-1}\cdots\xmd b_1 \xmd b_{0}$ in~$\Apqe$,
the \emph{incrementer}\\
\PushLine
$\addconst{w}=\aut{\N \times \{0,1,\ldots,n\},\Apq,\Apq,(0,0),\delta',\eta',\psi}$
\PushLine\\
is the (right) transducer
with input and output alphabet~$\Apq$, and whose
transition and output functions are defined by:
\begin{alignat}{4}
\forall s\in\N\quantvrg\forall a\in\Apq\quantvrg\e& \notag \\
\forall i<n\phantom{\quantvrg}\e &
(s,i)\pathaut{a|c}{}(s',i+1)& \e&\iff &
q\xmd s + (a+b_i) &= p\xmd s' + c  
\e
\notag
\\
& (s,n) \pathaut{a|c}{}(s',n)  & &\iff &
q\xmd s + a &= p\xmd s'+ c 
\notag
\end{alignat}
and whose
final function is defined by:
\begin{align}
\forall s\in\N\quantsp
\psi((s,n)) &= \pqRep{s} \EqVrgInt
\notag
\\
\psi((s,i)) &= \psi((s',i+1)).c
\e\text{ if }\e i<n
\e\text{and}\e (s,i) \path{0|c}(s',i+1)
\notag
\end{align}
\end{definition}

This last line means that if the input word is shorter than $w$, then
the final function behaves as if the input word ended with enough 0's (on the 
left, since we read from right to left).
\defin{inc-rem} describes a transducer with an infinite number of
states but, as in the case of the converter, it is easy to verify
that its reachable part is finite.
The incrementer $\addconst[{\base[3][2]}]{121}$ is shown at \figur{incr3125}.

\begin{figure}[ht]
\centering

\noindent\begin{subfigure}[b]{0.70\textwidth}
  \centering
    \scalebox{0.9}{\input{incr_3.125}}
\end{subfigure}\hfill%
\begin{subfigure}[b]{0.20\textwidth}
  \centering
    \scalebox{0.9}{\input{converter}}
\end{subfigure}

\caption{The incrementer $\addconst[{\base[3][2]}]{121}$}

\label{f.incr3125}
\end{figure}

It is  a simple verification that the incrementer has the
expected behaviour.

\begin{proposition}
	For every~$u$ and~$w$ in~$\Apqe$,
	$v=\addconst{w}(u)$ is a word in~$\Apqe$ such that
	$\val{v}=\val{u}+\val{w}$ holds.
\end{proposition}

\section{Proof of \protect\rtheorem{blip-fin-gen}}

The core of the proof lies in the next statement.

\begin{proposition}
\label{p.suc-cit}
For every~$w$ in~$\Apqe$, the image of a left-iterable language
by~$\addconst{w}$ is left-iterable.
\end{proposition}

\begin{proof}
Let~$u$ and~$v$ be in~$\Apqe$, $I\subseteq\N$ an infinite set of
indexes and $\{y_{i}\}_{i\in I}$ an infinite family of words
in~$\Apqe$.
The proof consists in showing that~$\Defi{\addconst{w}(u\xmd v^{i}y_{i})}{i\in I}$
is left-iterable.

Since~$I$ is infinite, we may assume, without loss of generality,
that the length of the~$y_{i}$'s is strictly increasing hence,
  that  all~$y_{i}$'s have a length greater \linebreak than~${n=\wlen{w}}$
but also that the reading of every~$y_{i}$ leads~$\addconst{w}$ to a
\emph{same} state~$(s,0)$:
\begin{equation}
\forall s\in\N\quantvrg\fa i\in I\quantsp
(0,n) \pathaut{y_{i}|y'_{i}}{\addconst{w}} (s,0)
\eqpnt
\eee\eee
\notag
\end{equation}
{}From the definition of the transitions of~$\addconst{w}$:
\begin{equation}
(s,0) \pathaut{a|c}{}(s',0)  \e\iff\e q\xmd s + a = p\xmd s'+ c
\eqvrg
\notag
\end{equation}
follows, since~$a<p$ and~$q<p$, that~$s\geq s'$.
Hence, the sequence of (first component of) states of~$\addconst{w}$ in a
computation starting in~$(s,0)$ and with input~$v^{i}$, with
unbounded~$i$, is ultimately stationary at state~$(t,0)$.

Without loss of generality, we thus may assume
that~$(0,n)\pathaut{y_{i}|y'_{i}}{}(t,0)$ for every~$i$ in~$I$ and,
since~$(t,0)\pathaut{v|v'}{}(t,0)$, it
holds that~$\addconst{w}(u\xmd v^{i}y_{i})=u'\xmd v'^{i}y'_{i}$,
where~$u'$ is the output of a computation starting in~$(t,0)$ and
with input~$u$.
\end{proof}

The special case of additive submonoids of~$\Vpq$ allows us to reverse the
condition from left-iterable to BLIP:
%

\begin{proposition}
\label{p.bli-ima-bli}
  Let~$w$ be a word of~$\Apqe$, and~$L$ be a BLIP language such 
  that~$\val{L}$
  is an additive submonoid of~$\Vpq$.
  The language $\addconst{w}(L)$ is BLIP.
\end{proposition}

\begin{proof}
  Since~$\val{L}$ is an additive submonoid of~$\Vpq$, it contains~$m\xmd\N$ for 
  some $m$ (as it must contains some number~$\frac{m}{q^l}$ for some~$m$ and~$l$).

Let~$n$ and~$k$ be the integers such that~$\val{w}=\frac{n}{q^k} = x$.
From \lemme{thr-esh}, it follows that there exists~$m_k$ such that for every~$j>m_k$, 
$\frac{j}{q_k}$ is in $\Vpq$. 
In particular, there exists~$j$ such that~$n+j \equiv 0 \mod (m\xmd q^k)$ 
and~$\frac{j}{q^k}$ is in~$\Vpq$. If we denote by~$y=\frac{j}{q^k}$, it means
that~$(x+y)$ is in~$m\xmd\N$. Hence, $\val{L}+x+y$ is contained in~$\val{L}$. 

Let us denote by~$u=\pqRep{y}$, and~$L'=\addconst{w}(L)$.

It follows that  $\val{\addconst{u}(L')} = (\val{L}+x+y) \subseteq{\val{L}}$, 
hence that~$\addconst{u}(L')$ is an
infinite subset of~$L$, and as such BLIP (from \rlemma{blip}).
If~$L'$ were left-iterable, so would be $\addconst{u}(L')$ by
\propo{suc-cit}, a contradiction.
\end{proof}

Finally we prove a property of finitely generated submonoids of~$\Vpq$.

\begin{proposition}
\label{p.fin-gen}
Let~$M$ be a finitely generated additive submonoid of~$\Vpq$.
There exists a finite family~$\{g_{i}\}_{i\in I}$ of elements
of~$\Vpq$ such
that~$M$ is contained \linebreak in~${\Cup_{i\in I} (g_{i}+\N)}$.
\end{proposition}

\begin{proof}

Let $\{y_1,y_2,\ldots,y_h\}$ be a generating family of~$M$.
Every $y_j$ is in~$\Vpq$ and it is then a rational
number~$\frac{n_j}{q^{k_j}}$  for some integers~$n_j$ and~$k_j$.
Let~$k$ be the largest of the~$k_{j}$.
Hence, every element in~$M$ is a rational number whose denominator is a
divisor of~$q^k$, and thus
$\msp M\subseteq \Vpq\cap\left(\frac{1}{q^k}\N\right)$.

Since every number in~$\frac{1}{q^k}\xmd\N$ can be
written as~$n+\frac{i}{q^k}$ for some $n$ in~$\N$
and some~$i$ in~$\set{{0},{1},\ldots,{q^k-1}}$,
it follows that $\frac{1}{q^k}\N = \Cup_{0\leq i < q^k} (\N +\frac{i}{q^k})$,
hence~${M\subseteq \Cup_{0\leq i < q^k} ( \Vpq\cap (\N + \frac{i}{q^k}))}$.
Besides, for every~$i$ in~$\set{{0},{1},\ldots,{q^k-1}}$, we denote by~$g_i$ 
the smallest number in~$\Vpq \cap (\N + \frac{i}{q^k})$.
Then, and since $\Vpq + \N=\Vpq$, for 
every~$i$,~$\Vpq\cap (\N + \frac{i}{q^k})= m_i+\N$.
Hence $M\subseteq \Cup_{0\leq i < q^k} (\N + m_i)$.%
\end{proof}

Even though this proposition seems rather weak (it is a poor
approximation from above), it is enough: it indeed reduces~\theor{blip-fin-gen}
to proving
that~$\cod{n+\N}$ (or equivalently~$\addconst{w}(\Lpq)$) is BLIP for
any~$n$, which was proven in \rproposition{bli-ima-bli}.

\begin{proof}[of \theor{blip-fin-gen}]
Let~$M$ be a finitely generated additive submonoid of~$\Vpq$.
By \propo{fin-gen}, there exists a finite family~$\{m_{i}\}_{i\in I}$
of elements of~$\Vpq$ such
that~$M \subseteq\Cup_{i\in I} (m_{i}+\N)$.

Let~$L=\pqRep{M}$ the language of the $\pq$-representations of the
elements of~$M$ and write~$w_{i}=\pqRep{m_{i}}$.
Hence, $L$ is contained in $(\Cup_{i}\addconst{w_i}(\Lpq))$, and thus
BLIP by \rlemma{blip}.
\end{proof}

\section{Conclusion and future work}

In this work, we have defined a new property, in an effort to capture
the structural complexity of~$\Lpq$.
This property contradicts any form of pumping lemma, placing~$\Lpq$
outside the scope of classical language theory.
Even more so that every other example of BLIP languages we describe seem 
to be purely artificial (\cf\rexample{blip})

Paradoxically, \theor{blip-fin-gen} shows that such examples are
very common within a rational base number system.
It seems that every reasonable number set is represented by a BLIP language
and that every simple language represents a complicated set of numbers.

%
%
%
%
\medskip

This work led us to a conjecture about rational
approximations of~$\Lpq$:

\begin{conjecture}\lproblem{lpq-approx}
  Let~$L$ be a rational language closed by addition and containing~$\Lpq$.
  Then~$L$ contains~$X.A_p^*$ where~$X=\Lpq\cap A_p^{\leq k}$, for some~$k$.
\end{conjecture}

Any approximation of~$\Lpq$ by a rational language~$L$,
would only keep a finite part of the structure:
the automaton accepting~$L$ would be the subtree of depth~$k$
of~$\Lpq$ whose leaves are all-accepting states.
\rfigure{conjecture} gives two examples of rational approximation of~$L_\frac{3}{2}$,
respectively when the~$\Lpq$ is cut at depth~$k=2$ and~$k=5$.

\begin{figure}
  \centering
  \scalebox{0.95}{\input{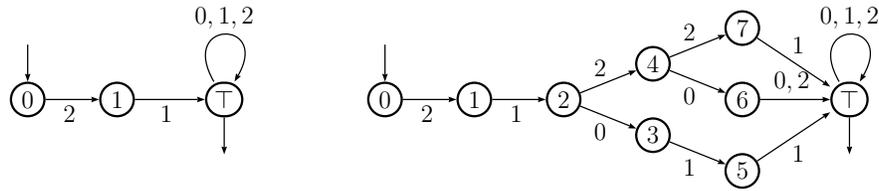}}
  \caption{Two rational approximations of~$L_{\frac{3}{2}}$}
  \lfigure{conjecture}
\end{figure}


%
%
%

\vfill
\bibliographystyle{plain}
\bibliography{%
  bibliography.bib,%
  Alexandrie-abbrevs,%
  Alexandrie-AC,%
  Alexandrie-DF,%
  Alexandrie-GL,%
  Alexandrie-MR,%
  Alexandrie-SZ%
}
\end{document}